\documentclass[12pt]{amsart}
\usepackage{fullpage,amsthm}
\usepackage{amsmath,amsfonts,amssymb} 
\usepackage{graphicx}
\pdfoutput=1

\input xy
\xyoption{all}

\newtheorem{thm}{Theorem}%[section]
\newtheorem{corl}[thm]{Corollary}
\newtheorem{lma}[thm]{Lemma} 
\newtheorem{prop}[thm]{Proposition}
\newtheorem{defn}[thm]{Definition}

\newtheorem{ex}[thm]{Example}
\newtheorem{rem}[thm]{Remark}
\def\A{\mathcal{A}}
\DeclareMathOperator{\ad}{ad}

\def\bar{\overline}

\def\C{\mathbb{C}}

\def\cO{\mathcal{O}}

\def\dd{\mathrm{d}}

\DeclareMathOperator{\End}{End}

\def\F{\mathbb{F}}

\def\g{\mathfrak{g}}
\def\gf{\mathrm{gf}}
\def\gh{\textup{gh}}
\def\H{\mathcal{H}}
\def\bH{\mathbb{H}}
\def\Hom{\mathrm{Hom}}
\def\half{\tfrac{1}{2}}

\def\k{\mathbf{k}}

\def\nn{\nonumber}
\def\o{\mathfrak{o}}

\DeclareMathOperator{\ord}{ord}

\def\R{\mathbb{R}}

\def\S{\mathcal{S}}

\def\sp{\mathfrak{sp}}
\def\su{\mathfrak{su}}

\def\SU{\mathcal{SU}}

\DeclareMathOperator{\tr}{Tr}
\def\tilde{\widetilde}

\def\U{\mathcal{U}}
\def\u{\mathfrak{u}}

\DeclareMathOperator{\Vol}{Vol}

\def\Z{\mathbb{Z}} 

\def\g{\mathfrak{g}}

\title[Renormalizability conditions for almost commutative manifolds]{Renormalizability conditions for \\almost commutative manifolds}
\author{Walter D. van Suijlekom}
\address{Institute for Mathematics, Astrophysics and Particle Physics,
Radboud University Nij-megen, Heyendaalseweg 135, 6525 AJ Nijmegen, The Netherlands}
\date{20 December 2011}

\begin{document}
\begin{abstract}
We formulate conditions under which the asymptotically expanded spectral action on an almost commutative manifold is renormalizable as a higher-derivative gauge theory. These conditions are of graph theoretical nature, involving the Krajewski diagrams that classify such manifolds. 
This generalizes our previous result on (super)renormalizability of the asymptotically expanded Yang--Mills spectral action to a more general class of particle physics models that can be described geometrically in terms of a noncommutative space. In particular, it shows that the asymptotically expanded spectral action which at lowest order gives the Standard Model of elementary particles is renormalizable. 

\end{abstract}

\maketitle

\section{Introduction}

Over the past few years it has turned out that many particle physics models can be described geometrically by modifying the internal structure of spacetime, making it slightly noncommutative. 
Indeed, there are so-called {\it almost commutative manifolds} that allow for a geometrical 
derivation of Yang--Mills theory \cite{CC97, BoeS10}, or even the full Standard Model, including Higgs potential and neutrino mass terms \cite{CCM07, CM07, JKSS07}. Theories that go beyond the Standard Model were described in \cite{Ste06,Ste07,SS07,JS08}.
Also supersymmetric models such as $N=1$ super-Yang--Mills theory and supersymmetric QCD have been derived geometrically \cite{BroS10,BroS11}. 

The basic idea in all these examples is that one describes an almost commutative manifold by spectral data, and then applies a general spectral action principle to derive physical Lagrangians. 
This paper continues on some of our recent results on renormalizability of the asymptotically expanded spectral action considered as higher-derivative theories: in \cite{Sui11b,Sui11c} we have shown that the Yang--Mills model is superrenormalizable as a gauge theory by observing that the asymptotically expanded spectral action contains natural higher-derivative regulators. We stress the importance of taking an asymptotic expansion, as it allows for a derivation of local Lagrangians, in contrast to eg. \cite{ILV11}. There, the full spectral action was considered as a non-local field theory, behaving completely differently for large momenta. 

In the present paper, we will formulate conditions for almost commutative manifolds that render the (asymptotically expanded) spectral action renormalizable as a gauge theory; even superrenormalizable in special cases. We show that these conditions apply to the aforementioned physical models. A convenient way to express our conditions is in terms of cycles in Krajewski diagrams \cite{Kra97} for the finite noncommutative geometries.

As we proceed, we note that the asymptotically expanded spectral action considered as a higher-derivative gauge theory is not multiplicatively renormalizable. This is in concordance with the interpretation of the spectral action as defining a physical theory at one particular mass scale, as already proposed in \cite{CC96,CC97}. For the Standard Model this mass scale is the GUT scale. 

This paper is organized as follows. In Section \ref{sect:ac} we recall some basic definitions from noncommutative geometry, specializing to almost noncommutative manifolds of the form $M \times F$: a product of an ordinary Riemannian manifold $M$ with a finite noncommutative space $F$. We recall Krajewski's diagrammatic classification and formulate a notion of R-connectedness; it will be related to renormalizability later on. We derive the gauge and scalar fields as a consequence of the noncommutative structure of $M \times F$. Essentially, these are coefficients for a twisted Dirac operator $D$. 

In Section \ref{sect:sa} we define the spectral action for $M \times F$ as
$$
\tr f (D/\Lambda)
$$
for some positive function $f$ and a cutoff parameter $\Lambda$. 
This is considered as an action functional in the gauge and scalar fields. 
We derive the lowest-order terms in an asymptotic expansion as $\Lambda \to \infty$, as well as the terms at any order in $\Lambda$ but quadratic in the fields. 

In Section \ref{sect:gf} we introduce a gauge fixing for almost commutative manifolds, much inspired by 't Hooft's $R_\xi$-gauge fixing for models with spontaneous symmetry breaking. This allows in Section \ref{sect:pc} for a power-counting argument to show that the asymptotically expanded spectral action on $M \times F$ is (super)renormalizable. Using results from the relevant BRST-cohomology, this is then completed to show renormalizability as a gauge theory, provided the Krajewski diagram for the finite space $F$ satisfies a certain graph-theoretical property, namely, the aforementioned R-connectedness. In particular, this applies to the asymptotically expanded spectral action that at lowest order is the Standard Model of elementary particles.

\section*{Acknowledgements}
The author would like to thank Matilde Marcolli for hospitality and Caltech for financial support during a visit in April 2011. The ESF is thanked for financial support under the program `Interactions of Low-Dimensional Topology and Geometry with Mathematical Physics (ITGP)'. NWO is acknowledged for support under VENI-project 639.031.827. Thijs van den Broek is gratefully acknowledged for a careful proofreading of the manuscript.

\section{Almost commutative manifolds}
\label{sect:ac}
The object of study in this paper is the {\it spectral action} \cite{CC96} for almost commutative manifolds. As a motivating example, let us start with a description of ordinary, commutative manifolds in terms of purely spectral data. Suppose $M$ is a compact Riemannian spin manifold, with a spinor bundle $\S$. Then the {\it Hilbert space} $L^2 (M,\S)$ of its square-integrable sections sets the stage for such a spectral description. The {\it Dirac operator} $D_M = i \gamma^\mu \circ \nabla^\S_\mu$ associated to the metric via the Levi--Civita connection $\nabla^\S$ lifted to the spinor bundle defines a self-adjoint operator on $L^2(M,\S)$. Ellipticity of $D_M$ as a differential operator and compactness of $M$ imply that the resolvents of $D_M$ are compact operators. Finally, the Dirac operator is compatible with the action of the coordinate functions: the action of functions $f \in C^\infty(M)$ on $L^2(M,\S)$ by pointwise multiplication has bounded commutators $[D_M, f]$.

If the manifold $M$ is of even dimension $m$, there is a grading (chirality) $\gamma_M$, making $D_M$ an odd operator. Finally, spin-manifolds are selected out of spin$^c$-manifolds by the charge conjugation operator: it is an anti-linear operator $J_M: L^2(M,\S) \to L^2(M,\S)$.

This canonical `triple' $(C^\infty(M), L^2(M,\S), D_M; \gamma_M, J_M)$ motivates the following abstract definition of a spectral triple \cite{C94}.

\begin{defn}
A {\rm spectral triple} $(\A,\H,D)$ is given by an unital $*$- algebra $\A$
represented faithfully as operators in a Hilbert space $\H$ and a self-adjoint
operator $D$ such that $(1+D^2)^{-1/2}$ is a compact operator and $[D,a]$ bounded for $a\in \A$.

\medskip

A spectral triple is {\rm even} if the Hilbert space $\H$ is
endowed with a $\Z/2\Z$-grading $\gamma$ such that $[\gamma,a]=0$ and $\{ \gamma,D\}=0$. 

\medskip

A {\rm real structure} of $KO$-dimension  $n \in \Z/8 \Z$ on a spectral triple is an antilinear isometry $J: \H \to \H$ such that
\begin{equation*}
J^2 = \varepsilon, \qquad JD = \varepsilon' DJ, \qquad J\gamma = \varepsilon'' \gamma J \quad \text{(even case)}.
\end{equation*}
The numbers $\varepsilon ,\varepsilon' ,\varepsilon'' \in \{ -1,1\}$
are a function of $n \mod 8$:
$$
\begin{array}
{|c| r r r r r r r r|} \hline n &0 &1 &2 &3 &4 &5 &6 &7 \\
\hline \hline
\varepsilon  &1 & 1&-1&-1&-1&-1& 1&1 \\
\varepsilon' &1 &-1&1 &1 &1 &-1& 1&1 \\
\varepsilon''&1 &{}&-1&{}&1 &{}&-1&{} \\  \hline
\end{array}
$$
Moreover, with $b^0 = J b^* J^{-1}$ we impose that
\begin{equation*}
[a,b^0] = 0, \qquad [[D,a],b^0] = 0, \qquad \forall \, a,b \in \A,
\end{equation*}
A spectral triple with a real structure is called a {\rm real spectral triple}.
\end{defn}
Thus, the real structure gives $\H$ the structure of an $\A$-bimodule. In other words, the algebra $\A \otimes \A^\circ$ acts on $\H$, where $\A^\circ$ is the opposite algebra to $\A$.

\begin{defn}
Let $(\A,\H,D)$ be a spectral triple. The $\A$-bimodule of {\rm Connes' differential one-forms} is given by
$$
\Omega^1_D(\A) := \left\{ \sum_k a_k[D,b_k]: a_k,b_k \in \A \right\}
$$
\end{defn}
In the case of the canonical triple, Clifford multiplication establishes an isomorphism (cf. \cite{C94, Lnd97})
$$
\Omega^1(M) \simeq \Omega^1_{D_M}(C^\infty(M)).
$$
%See \cite{C94, Lnd97} for more details.

Besides the canonical triple for a Riemannian spin manifold $M$, there is the following class of simple examples. 

\begin{defn}
A {\rm finite real spectral triple} is a spectral triple for which the Hilbert space is finite dimensional. We will write such a spectral triple suggestively, 
$$
F:= (A_F,H_F,D_F; \gamma_F, J_F)
$$
\end{defn}

\begin{ex}
\label{ex:ym-st}
The algebra $M_n(\C)$ of complex $n \times n$ matrices acts on itself by left and right matrix multiplication; this gives rise to a finite real spectral triple $$
(A_F= M_n(\C),H_F= M_n(\C),D_F=0; J_F=(\cdot)^*).
$$
This example is closely related to Yang--Mills theories (cf. \cite{CC96,CC97})
\end{ex}

\begin{ex}
\label{ex:sm-st}
The noncommutative description of the Standard Model is based on the real algebra 
$$
A_F = \C \oplus \bH \oplus M_3(\C).
$$
It is represented on $\C^{96}$, where 96 is 2 (particles and anti-particles) times 3 (families) times 4 leptons plus 4 quarks with 3 colors each. Finally, there is a $96 \times 96$ matrix $D_F$, a grading $\gamma_F$ and real structure $J_F$, which are explicitly described in \cite{CCM07, CM07}; they constitute a real spectral triple of KO-dimension $6$.
\end{ex}

We will be interested in a combination of Riemannian spin manifolds and such finite triples. 
\begin{defn}
An {\rm almost commutative manifold} (AC manifold) is given by the tensor product of the canonical triple and a finite spectral triple: 
$$
M \times F:= ( C^\infty(M)\otimes A_F, L^2(M,S) \otimes H_F, D_M \otimes 1 + \gamma_5 \otimes D_F)
$$
\end{defn}
The picture one should have in mind is that of Kaluza--Klein theories, where the spacetime manifold was extended by an extra dimension. In the present case, this extra dimension is a finite noncommutative space.

\subsection{Classification of finite spectral triples}
\label{sect:class}
We follow the work of Krajewski \cite{Kra97} where a diagrammatic way was introduced to classify finite real spectral triples, given by quintuples $(A_F,H_F,D_F;\gamma_F,J_F)$. We formulate these diagrams in slightly different terms. Recall that if $M$ is an $A_F$-bimodule, the {\it contragredient $A_F$-bimodule} $M^\circ$ is defined by
$$
M^\circ = \{ \bar m  : m \in M \}
$$
with action $\bar m \mapsto a \bar m b = \bar{b^* m a^*}$ for all $a,b \in A_F, m\in M$. In particular, if $M$ is a left $A_F$-module, $M^\circ$ is a right $A_F$-module.

The structure of $A_F$ can be determined explicitly from Wedderburn's Theorem:
\begin{equation}
\label{eq:wedderburn}
A_F \simeq \bigoplus_{i=1}^N M_{k_i}(\F_i).
\end{equation}
for some $k_1,\ldots, k_N$ and $\F_i = \R,\C$ or $\bH$ depending on $i$.

In the following, we will denote by $\Gamma^{(0)}$ and $\Gamma^{(1)}$ the vertex and edges sets of an oriented graph $\Gamma$ with source and target maps $s,t: \Gamma^{(1)} \to \Gamma^{(0)}$. Also, we indicate by $(v_1 v_2)$ an edge between vertices $v_1$ and $v_2$. 

%Denote by $\Bimod^s(A)$ the category of simple $A$-bimodules. If $A$ is finite-dimensional, then the objects of 

\begin{defn}
Given a finite-dimensional algebra $A_F = \bigoplus_{i} M_{k_i}(\F_i)$, a {\rm Krajewski diagram} for $A_F$ of KO-dimension $n$ is an oriented decorated graph $\Gamma$ with the following properties
\begin{enumerate}
\item Edges between two vertices come in pairs with opposite orientation: if $e = (v_1v_2)$ is an edge, then there also exists an edge $\bar e = (v_2 v_1)$ and these come in pairs. 
\item Each vertex $v$ is decorated by an irreducible $A_F$-bimodule $M_v$ together with a choice of basis, {\it i.e.} $M_v = \C^{n_v} \otimes \C^{m_v\circ}$ for some $n_v, m_v \in \{ k_1,\ldots, k_N\}$.
\item Each edge $e$ is decorated by a (non-zero) first-order operator $D_e: M_{s(e)} \to M_{t(e)}$, {\it i.e.} such that
$$
D_e ( a m b) = a D_e (m b) + D_e(a m)b - a D_e(m) b ; \qquad (a,b \in A_F, m \in M_v),
$$
and $D_{\bar e} = D_e^*$. 
\item There is an involutive graph automorphism $j: \Gamma \to \Gamma$ such that
$n_{j(v)} = m_v$ for all $v \in \Gamma^{(0)}$.
In other words, $M_{j(v)} = M_v^\circ$. 
If $J_v: M_v \to M_v^\circ$ is the anti-linear map that assigns to a bimodule its contragredient bimodule,\footnote{Actually, this is slightly more subtle in the case of KO-dimension 2,3,4, or 5; in that case, one needs two vertices $v_1, v_2$ and one has two maps $J_{v_1}: M_{v_1} \to M_{j(v_2)}$ and $J_{v_2}: M_{v_2} \to M_{j(v_1)}$ that satisfy $J_{j(v_2)} \circ J_{v_1}   = -1$. } 
 we demand that for an edge $e= (v_1v_2)$:
$$
D_{j(e)} =\epsilon'  J_{j(v_2)} D_e J_{j(v_1)}^{-1}
$$
\end{enumerate}
In the even case, there is an additional labeling on the vertices by signs $\pm 1$ and the edges connect only vertices of opposite signs. Moreover, if $v$ has sign $\pm 1$, then $j(v)$ has sign $\pm \epsilon''$. 
%$J_v$ maps the sign $\pm 1$ at $v$ to $\pm \epsilon''$ on $j(v)$. 
\end{defn}

Usually, one depicts a Krajewski diagram as embedded in the plane, with the columns and rows labeled by the integers $k_i$ that appear in the decomposition \eqref{eq:wedderburn} of $A_F$ . One places a node at $(n_v, m_v)$ for each vertex in $\Gamma$ with $M_v = \C^{n_v} \otimes \C^{m_v\circ}$. The pairs $(e,\bar e)$ of oriented edges in $\Gamma$ are indicated by a single line in the planar diagram, which by (3) run only horizontally or vertically. The graph automorphism $j$ translates as a reflectional symmetry of the diagram along the diagonal (with the labeling $\pm 1$ mapped to $\pm \epsilon''$). 
\begin{figure}[ht]
\begin{center}
\vspace{-10mm}
\begin{tabular}{c}
\begin{xy} 0;<3mm,0mm>:<0mm,3mm>::0;0,
%,(5,-2.5)*{\n_1}
,(5,-2.5)*{\cdots}
,(8,-2.5)*{\k_i}
,(11,-2.5)*{\cdots}
,(14,-2.5)*{\k_j}
,(17,-2.5)*{\cdots}
%,(23,-2.5)*{\n_N}
%,(1.5,-5)*{\n_1^\circ}
,(1.5,-5)*{\vdots}
,(1.5,-8)*{\k_i^\circ}
,(1.5,-11)*{\vdots}
,(1.5,-14)*{\k_j^\circ}
,(1.5,-17)*{\vdots}
%,(1.5,-23)*{\n_N^\circ}
,(8,-8)*\cir(0.3,0){}
,(14,-8)*\cir(0.3,0){}
,(8,-14)*\cir(0.3,0){}
,(8.2,-8);(13.8,-8)**\dir{-}
,(8,-8.2);(8,-13.8)**\dir{-}
\end{xy}
\end{tabular}
\caption{The lines between two nodes represent a non-zero $D_e: \C^{n_{s(e)}} \otimes \C^{m_{s(e)}^\circ} \to \C^{n_{t(e)}} \otimes \C^{m_{t(e)}^\circ}$, as well as its adjoint $D_{\bar e} :\C^{n_{t(e)}} \otimes \C^{m_{t(e)}^\circ} \to  \C^{n_{s(e)}} \otimes \C^{m_{s(e)}^\circ} $. The non-zero components $D_{j(e)}$ and $D_{j(\bar e)}$ are related to $\pm D_{e}$ and $ \epsilon' D_{\bar e}$.}
\end{center}
\end{figure}
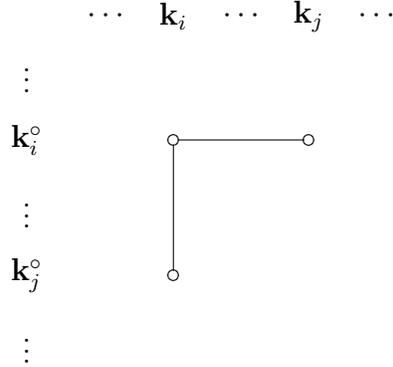

Given a Krajewski diagram $\Gamma = (\Gamma^{(0)}, \Gamma^{(1)})$ for $A_F$, we construct a finite spectral triple for the algebra $A_F$ as follows. We define
$$
H_F = \bigoplus_{v \in \Gamma^{(0)}} M_v =  \bigoplus_{v \in \Gamma^{(0)}} \C^{n_v} \otimes \C^{m_v \circ}
$$
on which $A_F$ acts on the left. The real structure $J_F$ is the sum of operators 
$$
J_v: M_v \to M_{j(v)}. 
$$
The Dirac operator $D_F$ is the sum of the operators
$$
D_e: M_{s(e)} \to M_{t(e)}.
$$
This defines a symmetric linear operator because $D_{\bar e} = D_e^*$. 
Finally, in the even case the signs on $M_v$ give rise to a grading $\gamma_F$ on $H_F$ for which $D_F$ is odd and the left action of $A_F$ on $H_F$ is even, and such that $\gamma_F J_F = \epsilon'' J_F \gamma_F$. 

\begin{ex}
\label{ex:kra-st}
The Krajewski diagram for the Standard Model is depicted in Figure \ref{kra-sm}. It indicates the precise structure of $H_F=\C^{96}$ as a representation space of $A_F = \C \oplus \bH \oplus M_3(\C)$. The double appearance of the row and column ${\bf 1}$ accounts for the multiplicities of the corresponding representations in $H_F$.
The nonempty blocks in the matrix $D_F$ are indicated by the dotted and straight lines, reflection along the diagonal gives $J_F$, and $\gamma_F$ is of opposite sign when reflected along the diagonal (KO-dimension 6). 
\begin{figure}[ht]
\begin{center}
\vspace{-10mm}
\begin{tabular}{c}
\begin{xy} 0;<3mm,0mm>:<0mm,3mm>::0;0,
,(5,-2.5)*{{\bf 1}}
,(10,-2.5)*{{\bf 2}}
,(15,-2.5)*{{\bf 1}}
,(20,-2.5)*{{\bf \bar 1}}
,(25,-2.5)*{{\bf 3}}
,(1.5,-5)*{{\bf 1}^\circ}
,(1.5,-10)*{{\bf 2}^\circ}
,(1.5,-15)*{{\bf 1}^\circ}
,(1.5,-20)*{{\bf \bar 1}^\circ}
,(1.5,-25)*{{\bf 3}^\circ}
,(10,-5)*\cir(0.3,0){}
,(15,-5)*\cir(0.3,0){}
,(20,-5)*\cir(0.3,0){}
,(10.2,-5);(14.8,-5)**\dir{-}
%,(17.2,-5);(22.8,-5)**\dir{-}
,(10.1,-4.9);(19.9,-4.9)**\crv{(15,-3.5)}
,(10,-25)*\cir(0.3,0){}
,(15,-25)*\cir(0.3,0){}
,(20,-25)*\cir(0.3,0){}
,(10.2,-25);(14.8,-25)**\dir{-}
%,(17.2,-29);(22.8,-29)**\dir{-}
,(10.1,-25.1);(19.9,-25.1)**\crv{(15,-26.5)}
,(5,-10)*\cir(0.3,0){}
,(5,-15)*\cir(0.3,0){}
,(5,-20)*\cir(0.3,0){}
,(5,-10.2);(5,-14.8)**\dir{-}
%,(5,-17.2);(5,-22.8)**\dir{-}
,(4.9,-10.1);(4.9,-19.9)**\crv{(3.5,-15)}
,(25,-10)*\cir(0.3,0){}
,(25,-15)*\cir(0.3,0){}
,(25,-20)*\cir(0.3,0){}
,(25,-10.2);(25,-14.8)**\dir{-}
%,(29,-17.2);(29,-22.8)**\dir{-}
,(25.1,-10.1);(25.1,-19.9)**\crv{(26.5,-15)}
,(5.1,-14.9);(14.9,-5.1)**\dir{--}
\end{xy}
\end{tabular}
\caption{The Krajewski diagram of the Standard Model}
\label{kra-sm}
\end{center}
\end{figure}

\end{ex}

Krajewski has shown in \cite{Kra97} that there is a one-to-one correspondence between such diagrams and finite real spectral triples modulo unitary equivalence (related to the choice of basis on $M_v$). We recall from {\it loc.~cit.}~a useful result for the Dirac operator. 
\begin{lma}[\cite{Kra97}]
\label{lma:splitting-DF}
Let $(A_F,H_F,D_F; \gamma_F,J_F)$ be a finite real spectral triple. There is a decomposition $D_F= D_0 + \Delta + J_F \Delta J_F^{-1}$ where $D_0$ commutes with the action of $A_F \otimes A_F^\circ$ and $\Delta$ commutes with the action of $A_F^\circ$. Moreover, 
\begin{equation}
\label{eq:int-form}
\Delta = - \int_{\U(A_F)} g [\Delta, g^*] d \mu(g).
\end{equation}
where $d\mu$ is the Haar measure on the Lie group $\U(A_F)$ consisting of unitaries in $A_F$.
\end{lma}
\proof
If $\Gamma$ is the Krajewski diagram corresponding to $(A_F,H_F,D_F; \gamma_F,J_F)$ then we can decompose $D= \sum_e D_e$ with each $D_e: M_{s(e)} \to M_{t(e)}$. As before, write $M_v= \C^{n_v} \otimes \C^{m_v\circ}$ for all $v$. 
We denote by $D_0$ the sum of such operators for which $n_{s(e)} = n_{t(e)}$ and $m_{s(e)}= m_{t(e)}$. If this is not the case, then since any first-order operator such as $D_e$ splits into left $A_F$-linear and right $A_F$-linear components we have that either $n_{s(e)}=n_{t(e)}$ or $m_{s(e)} =m_{t(e)}$, respectively. It is clear that $D_F-D_0$ is the sum of these, so we define the right $A_F$-linear (or left $A_F^\circ$-linear) operator
$$
\Delta =  \sum_{e: m_{s(e)} =m_{t(e)}} D_e
$$
Then $J_F \Delta J_F^{-1}$ gives the remaining sum over all edges $e$ for which $n_{s(e)}=n_{t(e)}$, showing $ D_F= D_0 + \Delta + J_F \Delta J_F^{-1}$.

For the integral formula \eqref{eq:int-form}, we compute the matrix coefficients of the difference between the left and the right-hand side in Eq. \eqref{eq:int-form} in terms of a basis of the Hilbert space $H_F$. We decompose 
$$
H_F= \bigoplus_v M_v = \bigoplus_v \C^{n_v} \otimes \C^{m_v\circ} ; \qquad (\alpha = 1, \ldots,n_v,\beta =1 ,\ldots m_v)
$$
and write  $\{ e_v^{\alpha\beta} \}$ for the corresponding basis. Then,
$$
\left\langle e_{v_1}^{\alpha_1 \beta_1}, \int_{\U(A_F)} g \Delta g^* d \mu(g) e_{v_2}^{\alpha_2 \beta_2} \right\rangle  
= \int_{\U(A_F)}  \left\langle e_{v_1}^{\alpha_1 \beta_1}, g e_{w_1}^{\alpha_1'\beta_1'} \right\rangle  \Delta_{(w_2w_1)}^{\alpha_1'\alpha'_2}\delta^{\beta_1'\beta_2'}  \left\langle e_{w_2}^{\alpha_2'\beta'_2},  g^*  e_{v_2}^{\alpha_2 \beta_2} \right\rangle d \mu(g)
$$
where we sum over all repeated indices and where $\Delta_{(w_2w_1)}^{\alpha_1'\alpha'_2} \delta^{\beta_1'\beta_2'}$ are the matrix coefficients of (the right $A_F$-linear) $\Delta_{(w_2w_1)}: M_{w_2} \to M_{w_1}$. Next, 
$$
\left\langle e_{v_1}^{\alpha_1 \beta_1}, g e_{w_1}^{\alpha_1'\beta'_1} \right\rangle = \delta_{v_1,w_1} g^{\alpha_1\alpha_1'}_{v_1} \delta^{\beta_1\beta_1'}
$$
in terms of the {\it defining} matrix coefficients $g^{\alpha_1\alpha_1'}_{v_1}$ of $g \in \U(A_F)$ in the representation $\C^{n_{v_1}}$. Note that this is a representation of $\U(A_F)$, since $\U(A_F) \simeq \prod_i U(k_i, \F_i)$. 
This turns the above integral into
$$
\Delta_{(v_2 v_1)}^{\alpha_1'\alpha'_2} \delta^{\beta_1 \beta_2} \int_{\U(A_F)} g^{\alpha_1\alpha_1'}_{v_1} 
   \bar{g^{\alpha_2\alpha_2'}_{v_2}} d \mu(g)
= \Delta_{(v_2 v_1)}^{\alpha_1\alpha_2} \delta^{\beta_1 \beta_2} \delta_{v_1, v_2} 
$$
by the Peter--Weyl Theorem. However, $\Delta_{(v_2v_1)}$ maps between different irreducible representations of $\U(A_F)$ which implies the vanishing of the above expression and completes the proof.
\endproof

We will now formulate a condition on Krajewski diagrams that below will turn out to characterize renormalizable models. Let $\Gamma$ be a Krajewski diagram for $A_F$. We construct a graph $\tilde \Gamma$ whose vertex set $\tilde \Gamma^{(0)}$ is the set of inequivalent irreducible representations of $A_F$, or, of the Lie group $\U(A_F)$. In other words, $\tilde \Gamma^{(0)}$ is the set $\{ k_1, \ldots, k_N\}$ appearing in the decomposition \eqref{eq:wedderburn}. The set of edges for $\tilde \Gamma$ is defined as
$$
\tilde \Gamma^{(1)} := \{ (n ,n') : \exists e \in \Gamma^{(1)} \text{ such that } n_{s(e)} = n \text{ and }  n_{t(e)} = n' \}
$$
There is a map of graphs $\psi: \Gamma \to \tilde \Gamma$ defined as follows. For a vertex $v \in \Gamma^{(0)}$ we set $\psi(v) = n_v \in \tilde \Gamma^{(0)}$; for an edge $e \in \Gamma^{(1)}$ we set 
$$
\psi(e) = (n_{s(e)}, n_{t(e)})
%\left\{ \begin{array}{ll} (n_{s(e)}, n_{t(e)}) & \text{ if } m_{s(e)}= m_{t(e)}\\(m_{s(e)}, m_{t(e)}) & \text{ if } n_{s(e)}= n_{t(e)}\end{array} \right.
$$
Essentially, the graph $\tilde \Gamma$ is the projection of the Krajewski diagram $\Gamma$ onto the horizontal axis. By symmetry, $\tilde \Gamma$ is also the projection of $\Gamma$ onto the vertical axis; this corresponds to pre-composing $\psi$ with the graph automorphism $j: \Gamma \to \Gamma$.
 
Adopting the usual terminology from graph theory, we will refer to an edge with the same source and target vertex as a {\it loop}; a {\it cycle} is a path which begins and starts at the same vertex, but with no other repeated vertices ({\it i.e.} it does not contain loops). 

\begin{defn}
In the above notation, a {\rm lift along $\psi$} of a cycle $\gamma= \tilde e_1 \cdots  \tilde e_m$ of length $m$ in $\tilde \Gamma$ is a cycle $e_1 \cdots e_l$ of length $l\geq m$ such that the path $\psi(e_1) \cdots \psi(e_l)$ coincides with $\tilde \gamma$ modulo loops. 
\end{defn}
Figure \ref{fig:hor-lift} illustrates such a `horizontal lift' of graphs; similarly, we can define a vertical lift by using the map $\psi \circ j$.

%Note that such edges $e_i$ are always horizontal, but can be related to vertical edges via the graph automorphism $j$; in fact, we could have defined a lift along $\psi \circ j$ to obtain such a collection of vertical edges.
%In general, any cycle $\gamma$ in $\Gamma$ can be decomposed in a horizontal and vertical part, we write $\gamma= \gamma_H \ast \gamma_V$. Note that the individual $\gamma_H$ and $\gamma_V$ are just sets of edges and not necessarily cycles in $\Gamma$, as illustrated by the following example:
\begin{figure}[h!]
\begin{center}
\vspace{-15mm}
\begin{tabular}{c}
\begin{xy} 0;<3mm,0mm>:<0mm,3mm>::0;0,
,(5,-5)*\cir(0.3,0){}
,(10,-5)*\cir(0.3,0){}
,(5,-10)*\cir(0.3,0){}
,(10,-10)*\cir(0.3,0){}
,(5.2,-5);(9.8,-5)**\dir{-} ?>*\dir{>}
,(5.2,-10);(9.8,-10)**\dir{-}?<*\dir{<}
,(5,-5.2);(5,-9.8)**\dir{-} ?<*\dir{<}
,(10,-5.2
);(10,-9.8)**\dir{-} ?>*\dir{>}
,(15,-8.5)*+{\underset{\psi}{\leadsto}}
,(20,-7.5)*\cir(0.3,0){}
,(25,-7.5)*\cir(0.3,0){}
,(20.2,-7.5);(24.8,-7.5)**\crv{(22.5,-6)} ?>*\dir{>}
,(20.2,-7.5);(24.8,-7.5)**\crv{(22.5,-9)} ?<*\dir{<}
\end{xy}
\end{tabular}
\end{center}
\caption{The cycle in $\Gamma$ at the left-hand side is a lift along $\psi$ of the cycle in $\tilde \Gamma$ at the right-hand side (we have suppressed the loops at the two vertices).}
\label{fig:hor-lift}
\end{figure}
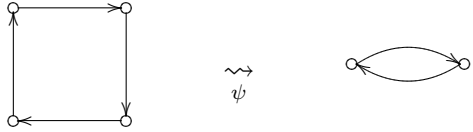

\begin{defn}
\label{defn:prop-R}
We say that a Krajewski diagram $\Gamma$ is {\rm R-connected in dimension $m$} if 
\begin{enumerate}
\item every cycle $\tilde \gamma$ in $\tilde \Gamma$ of length less than $m$ can be lifted along $\psi$ to a cycle $\gamma$ in $\Gamma$,
\item every two cycles $\tilde \gamma_1,  \tilde \gamma_2$ of total length less than $m$, which are not connected to a common vertex $1$ or $\bar 1$ in $\tilde \Gamma$, can be lifted to a single cycle $\gamma$ in $\Gamma$ along $\psi$ and $\psi \circ j$, respectively, {\it i.e.}
$$
\psi(\gamma) \sim \tilde\gamma_1; \qquad \psi(j(\gamma)) \sim \tilde \gamma_2
$$
where $\sim$ denotes equivalence of cycles in $\tilde \Gamma$ modulo loops.
%, such that $\gamma = \gamma_H \ast \gamma_V$ is a cycle in $\Gamma$, 
%\item for every pair of cycles $\tilde \gamma_1,  \tilde \gamma_2$ of total length $m$ connected to a common vertex $v = 1$ or $\bar 1$ in $\tilde \Gamma$, the concatenation $\tilde \gamma_v \ast_1 \tilde \gamma_2$ at that common vertex $v$ can be lifted to a cycle $\gamma$ in $\Gamma$ as in (1),
\item For $r \geq 3$, there are no tuples $\tilde \gamma_1, \ldots \tilde \gamma_r$ of cycles of total length $m$, which are not mutually connected to a common $1$ or $\bar 1$. 
\end{enumerate}
\end{defn}
Note that the last condition is trivially satisfied in the case $m \leq 4$, since every cycle has length at least 2. The case $m=4$ happens to be our case of interest.

\begin{prop}
\label{prop:kra-sm}
The Krajewski diagram of the Standard Model (Figure \ref{kra-sm}) is R-connected in dimension 4. 
\end{prop}
\proof
Indeed, the graph $\tilde \Gamma$ is given by
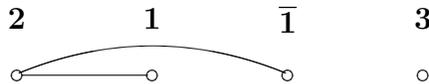
\begin{center}
\vspace{-10mm}
\begin{tabular}{c}
\begin{xy} 0;<3mm,0mm>:<0mm,3mm>::0;0,
,(11,-2.5)*{{\bf 2}}
,(17,-2.5)*{{\bf 1}}
,(23,-2.5)*{{\bf \bar 1}}
,(29,-2.5)*{{\bf 3}}
,(11,-5)*\cir(0.3,0){}
,(17,-5)*\cir(0.3,0){}
,(23,-5)*\cir(0.3,0){}
,(29,-5)*\cir(0.3,0){}
,(11.2,-5);(16.8,-5)**\dir{-}
%,(17.2,-5);(22.8,-5)**\dir{-}
,(11.1,-4.9);(22.9,-4.9)**\crv{(17,-2.5)}
\end{xy}
\end{tabular}
\end{center}
where we have suppressed the loops. Every cycle and every pair of distinguished cycles in $\tilde \Gamma$ of total length 4 can be lifted to a single cycle $\gamma$ in the Krajewski diagram of Figure \ref{kra-sm}. The pair consisting of two copies of the cycle $({\bf 12})({\bf 21})$ (going back-and-forth between ${\bf 1}$ and ${\bf 2}$) have a common vertex ${\bf 1}$. For this reason, they do not enter in Condition (2) (which, in fact, they would not satisfy). The concatenated cycle $({\bf 12})({\bf 21})({\bf 12})({\bf 21})$ (going back-and-forth twice between ${\bf 1}$ and ${\bf 2}$) is of length 4, and was already treated (cf. Condition (1)). A similar argument applies to the cycle $({\bf \bar 12})({\bf 2 \bar 1 })$.
\endproof

\begin{ex}
\label{ex:not-R}
Let us give an example of a Krajewski diagram which is not R-connected (in dimension 4). Consider 
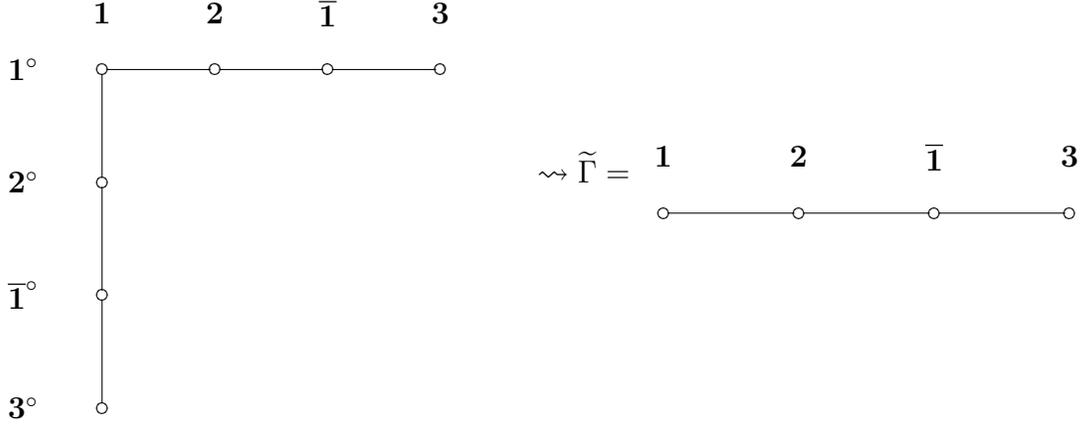
\begin{center}
\vspace{-10mm}
%$\Gamma = \qquad $
\hspace{-5mm}\begin{tabular}{c}
\begin{xy} 0;<3mm,0mm>:<0mm,3mm>::0;0,
,(5,-2.5)*{{\bf 1}}
,(10,-2.5)*{{\bf 2}}
,(15,-2.5)*{{\bf \bar 1}}
,(20,-2.5)*{{\bf 3}}
,(1.5,-5)*{{\bf 1}^\circ}
,(1.5,-10)*{{\bf 2}^\circ}
,(1.5,-15)*{{\bf \bar 1}^\circ}
,(1.5,-20)*{{\bf 3}^\circ}
,(5,-5)*\cir(0.3,0){}
,(10,-5)*\cir(0.3,0){}
,(15,-5)*\cir(0.3,0){}
,(20,-5)*\cir(0.3,0){}
,(5.2,-5);(9.8,-5)**\dir{-}
,(10.2,-5);(14.8,-5)**\dir{-}
,(15.2,-5);(19.8,-5)**\dir{-}
,(5,-10)*\cir(0.3,0){}
,(5,-15)*\cir(0.3,0){}
,(5,-20)*\cir(0.3,0){}
,(5,-5.2);(5,-9.8)**\dir{-}
,(5,-10.2);(5,-14.8)**\dir{-}
,(5,-15.2);(5,-19.8)**\dir{-}
\end{xy}
\end{tabular}
$\qquad \leadsto \tilde \Gamma = $  
\begin{tabular}{c}
\begin{xy} 0;<3mm,0mm>:<0mm,3mm>::0;0,
,(5,0)*{{\bf 1}}
,(11,0)*{{\bf 2}}
,(17,0)*{{\bf \bar 1}}
,(23,0)*{{\bf 3}}
,(5,-2.5)*\cir(0.3,0){}
,(11,-2.5)*\cir(0.3,0){}
,(17,-2.5)*\cir(0.3,0){}
,(23,-2.5)*\cir(0.3,0){}
,(5.2,-2.5);(10.8,-2.5)**\dir{-}
,(11.2,-2.5);(16.8,-2.5)**\dir{-}
,(17.2,-2.5);(22.8,-2.5)**\dir{-}
\end{xy}
\end{tabular}
\end{center}
Indeed, the pair of cycles $\{ ({\bf 12} )({\bf 21} ), ({\bf \bar 1 3} )({\bf 3 \bar 1} )\}$ obtained by going back and forth along the left and right edge in $\tilde \Gamma$ does not lift to a cycle in $\Gamma$.
\end{ex}

\subsection{Gauge fields from AC manifolds}
Let us now describe how noncommutative manifolds naturally give rise to a gauge theory (cf. \cite[Section 10.8]{CM07}). 
For simplicity, we will restrict to almost commutative manifolds, so that $(\A,\H,D)$ will always denote $M \times F$, {\it i.e.}
$$
(\A,\H,D)= (C^\infty(M, A_F), L^2(M,\S) \otimes H_F, D_M \otimes 1 + \gamma_M \otimes D_F).
$$
for some finite spectral triple $F = (A_F,H_F,D_F;\gamma_F,J_F)$.

\begin{defn}
Denote by $\U(\A)$ the group of unitaries of $\A$. The {\rm gauge group} of $M \times F$ is given by
$$
\SU(\A) = \{  u \in \U(\A): {\det}_{F} u = 1\}
$$
where the determinant is taken pointwise in the representation $H_F$.
\end{defn}
The group $\SU(\A)$ acts naturally on the Dirac operator $D$ by conjugation, as well as on the representation of $\A$ on $\H$: $a \mapsto u a u^*$. If there is a real structure, then we transform 
$$
D \mapsto U D U^*,
$$ 
with $U = u u^{*\circ} \equiv u J u J^{-1}$. This suggests that we should rather take the image of $\SU(\A)$ under the map $u \mapsto u u^{* \circ}$. Indeed, in \cite{DS11} the gauge group was defined in this way, leading to a quotient of $\SU(\A)$ by an abelian group. Since in most of our examples the latter group will be finite, it will be ignored in what follows.

\begin{prop}
\label{prop:su}
Let $M \times F$ be an almost commutative manifold as above and write
$$
A_F = \bigoplus_{i=1}^N M_{k_i}(\F_i); \qquad (\F_i = \R,\C \text{ or }\bH).
$$
\begin{enumerate}
\item The gauge group of $M \times F$ is given by $\SU(\A) = C^\infty(M, \SU(A_F))$. 
\item The Lie algebra $\su(A_F)$ of $\SU(A_F)$ is isomorphic to
$$
\su(\A) \simeq  \bigoplus_{i=1}^N \su(k_i) \oplus \u(1) ^{\oplus(C-1)}
$$
where $C$ is the number of complex algebras in the above decomposition of $A_F$ and $\su(k_i)$ denotes $\o(k_i), \su(k_i)$ or $\sp(k_i)$ depending on whether $\F_i = \R,\C$ or $\H$, respectively.
\end{enumerate}
Consequently, there is a one-to-one correspondence between irreducible representations of the algebra $A_F$ and of the Lie algebra $\su(A_F)$ provided $A_F$ contains no copies of $\R$, and either no complex subalgebras, or at least one non-trivial (i.e. not $\C$) complex subalgebra. 
\end{prop}
\proof
(1) is direct. (2). Note that $\u(A)$ is a direct sum of simple Lie algebras $\o(k_i), \u(k_i)$ and $\sp(k_i)$ according to $\F_i = \R,\C,\bH$, respectively. All these matrix Lie algebras have a trace, and we observe that the matrices in $\o(k_i)$ and $\sp(k_i)$ are already traceless. For the complex case, we can write $X_i \in \u(k_i)$ as $X_i= Y_i+z_i$ where $z_i = \tr X_i$, showing that:
$$
\u(k_i) = \su(k_i) \oplus \u(1).
$$
The determinant condition in the definition of $\SU(\A)$ translates at the infinitesimal level to the unimodularity condition $\tr_{H_F} X = 0$. Explicitly, this becomes
$$
\sum_i \alpha_i \tr (X_i) = 0
$$
where $\alpha_i$ are the multiplicities of the fundamental representations of $M_{k_i}(\F_i)$ appearing in $H_F$. Using the above property for the traces on simple matrix Lie algebras, we find that unimodularity is equivalent to
$$
\sum_{l=1}^C \alpha_{i_l} z_{i_l} = 0 
$$
where the sum is over the complex factors ({\it i.e.} for which $\F_i = \C$) in $A$, labeled by $i_1, \ldots, i_C$. This reduces the $C$ abelian factors to $C-1$ copies of $\u(1)$.
\endproof

\begin{ex}
For the Standard Model spectral triple of Example \ref{ex:sm-st} (cf. Example \ref{ex:kra-st}) this gives $\su(A_F) = \su(3) \oplus \su(2) \oplus \u(1)$, as desired (\cite[Proposition 2.16]{CCM07} or \cite[Proposition 1.185]{CM07}). 
\end{ex}

Now that we have found the gauge group of an almost commutative manifold, let us determine the gauge fields that $M \times F$ naturally gives rise to through the differential one-forms. 

\begin{prop}
\label{prop:splitting-omega}
The differential one-forms $\Omega^1_{D_M}(\A)$ on $M \times F$ allow for a direct sum decomposition:
$$
\Omega^1_D(\A) \simeq \Omega^1(M, A_F) \oplus C^\infty(M, \Omega^1_{D_F}(A_F)).
$$
where $\Omega^1(M, A_F) \equiv \Omega^1(M) \otimes A_F$. Moreover, the $A_F$-bimodule of differential one-forms $\Omega^1_{D_F}(A_F)$ is generated by $\Delta$. 
\end{prop}
\proof 
This follows directly from the splitting
$$
D = D_M \otimes 1 + \gamma_M \otimes D_F
$$
noting further that $\gamma_\mu$ and $\gamma_M$ are orthogonal with respect to the Hilbert--Schmidt inner product. 

The integral formula for $\Delta$ in Lemma \ref{lma:splitting-DF} combined with the observation that $[D,a] = [\Delta,a]$ for all $a \in A_F$ shows that $\Delta$ is already a one-form; this shows that $A_F \Delta A_F \subset \Omega^1_{D_F}(A_F)$. The same observation also shows that $\Omega^1_{D_F}(A_F) \subset A_F \Delta A_F$.
\endproof

Let us describe the linearly independent components of $\Omega^1_{D_F}(A_F)$; inspired by the discussion in Krajewski \cite{Kra97}. %Recall the construction of the graph $\tilde \Gamma$ from a Krajewski diagram $\Gamma$ in the previous section.

An element $\phi \in \Omega^1_{D_F}(A_F)$ is given by sums of elements of the form
$$
a \Delta b =  \sum_{e: m_{s(e)}=m_{t(e)}} a D_e b.
$$
Since some edges induce linear operators $D_e$ between the same representations of $A_F$, the above summands are not independent. In order to turn this into a sum over linearly independent terms, the graph $\tilde\Gamma$ introduced previously is quite convenient. Namely, given an edge $\tilde e$ in $\tilde \Gamma$ connecting different vertices, we consider the linear span $S_{\tilde e}$ in $\Hom(\C^{s(\tilde e)}, \C^{t(\tilde e)})$ of all matrices $D_e$ with $e \in \psi^{-1}(\tilde e)$. 
% for which $V_{s(e)} = \tilde V_1$ and $V_{t(e)} = \tilde V_2$. 
If $\{f_{\tilde e}^p \}_p$ ($p=1 ,\ldots, S_{\tilde e}$) is a basis for $S_{\tilde e}$ we can write 
\begin{equation}
\label{coeff-De}
D_e = \sum_p M_e^p f_{\psi(e)}^p, \qquad M_e^p \in \C.
\end{equation}
%where $\tilde e$ is the edge in $\tilde \Gamma$ corresponding to $e$ in $\Gamma$.
Note that the self-adjointness of $D$ implies that $M_{\bar e}^p = \bar M_e^p$ and $f_{\psi(\bar e)}^p = (f_{\psi(e)}^p)^*$. 
Adopting this form of $D_e$, we can write 
$$
a \Delta b = \sum_{\tilde e \in \tilde  \Gamma^{(1)}: s(\tilde e) \neq t(\tilde e)}
\sum_{p} \left(\sum_{\begin{smallmatrix} e \in \psi^{-1}(\tilde e)\\ m_{s(e)} = m_{t(e)} \end{smallmatrix}} M_e^p \right)
a f_{\tilde e}^p b.
$$
%The primed sum is further restricted over all edges $e$ in $\Gamma$ for which $m_{s(e)} =m_{t(e)}$, as these appear in the definition of $\Delta$. %, $V_{s(e)} \simeq \tilde V_1 $ and $V_{t(e)} \simeq \tilde V_2$, {\it i.e.} such that $\tilde e$ corresponds to $e$. 
We denote the independent fields by
$$
\phi^p_{\tilde e} = a f_{\tilde e}^p b : \qquad a,b \in A_F.
$$
which is an element in $\Hom(\C^{s(\tilde e)}, \C^{t(\tilde e)})$. Thus, we can write a general element $\phi \in \Omega^1_{D_F}(A_F)$ as
$$
\phi= \sum_{e: m_{s(e)}= m_{t(e)}} \sum_p M_e^p \phi_{\tilde e}^p.
$$
We conclude that the number of independent components of $\phi$ is
$$
\sum_{\tilde e: s(\tilde e) \neq t(\tilde e)} s(\tilde e) t(\tilde e) \dim S_{\tilde e} 
$$
A corresponding orthonormal basis (orthonormal with respect to the Hilbert--Schmidt norm on $\Omega^1_{D_F}(A_F) \subset \End H_F$) can be found by combining the indices $\tilde e$ and $p$ with the canonical bases of $\C^{s(\tilde e)}$ and $\C^{t(\tilde e)}$: we denote this orthonormal basis of $\Omega^1_{D_F}(\A_F)$ by $\{ e_I \}_I$. 

The vertices of $\tilde \Gamma$ label irreducible representations of $A_F$, and consequently of $\su(A_F)$.  Thus, the fields $\phi^p_{\tilde e}$  carry the induced representation, that is, by conjugation of the source and target representations $s(\tilde e)$ and $t(\tilde e)$. 

\begin{ex}

Consider the Krajewski diagram of the Standard Model (Figure \ref{kra-sm}). The fields that appear connect the vertices $2$ and $1$, and $2$ and $\bar 1$ in $\tilde \Gamma$: they carry the induced representation of $\u(1) \oplus \su(2)$. In fact, this is precisely the Higgs doublet in the electroweak model, having 2 independent degrees of freedom.
\end{ex}

Let us end this section by describing the inner fluctuations of the metric, induced by coupling $D$ to gauge fields in $\Omega^1_D(\A)$. 
The origin of this can also nicely be described in terms of Morita self-equivalences of the algebra $\A$ (cf. \cite[Sect. 10.8]{CM07}).

We consider a self-adjoint element $\omega +\gamma_M \phi \in \Omega^1_D(\A)$, in terms of the splitting in Proposition \ref{prop:splitting-omega}. The unimodularity condition on the gauge group is transferred to the gauge fields by demanding that $\tr_F \omega = 0$. Combining this with self-adjointness implies that $\omega \in \Omega^1(M, i\su(A_F))$. This allows for an inner fluctuation:
$$
D \leadsto D+ A + \gamma_M \Phi
$$
where 
$$
A = \omega + \epsilon'  J \omega J^{-1} = i \gamma^\mu \ad \omega_\mu ; \qquad \Phi = \phi + \epsilon' J \phi J^{-1}.
$$
These formulas can be checked using the splitting of Proposition \ref{prop:splitting-omega}. 

\begin{rem}
\label{rem:traces}
Note that a term such as $\tr_F \Phi^n$ (with the trace over $H_F$) can be easily computed from the Krajewski diagram. Indeed, it corresponds to a sum over cycles in $\Gamma$ of length $l$, for which the trace splits over a horizontal and vertical part:
$$
\tr_F \Phi^n = \sum_{\gamma= e_l \cdots e_1} \sum_{p_i} c_{p_1\ldots p_l}(\gamma) 
\tr_{s(\tilde e_{1})} \left( \phi^{p_l}_{\tilde e_l} \ldots \phi^{p_{1}}_{\tilde e_{1}} \right) 
\tr_{s(\tilde{ j(e_1)})} \left(\phi^{p_l}_{\tilde{j( e_l)}} \ldots \phi^{p_{1}}_{\tilde{ j(e_{1})}} \right)
$$
where we have denoted $\tilde e_i = \psi(e_i)$ and $\phi_{\tilde e}^p \equiv 1$ if $\tilde e$ is a loop in $\tilde \Gamma$ (i.e., if $s(\tilde e) = t(\tilde e)$).
%Essentially, we have split $\gamma$ into a horizontal and vertical part. 
The coefficient is given essentially by 
$$
c_{p_1\ldots p_n} (\gamma) \propto M_{e_1}^{p_1} \cdots M_{e_l}^{p_{l}}
$$
in terms of the basis coefficients of $D_e$ in Equation \eqref{coeff-De}. 
%If the Krajewski diagram is R-connected of Definition \ref{defn:prop-R}, then this can be reduced to a sum over loops $\tilde \gamma$ in $\tilde \Gamma$:
%$$
%\tr_F \Phi^n = \sum_{\tilde \gamma = \tilde e_1 \cdots \tilde e_n} \sum_{p_i} c'_{p_1\ldots p_n}(\tilde \gamma) \tr_{s(\tilde e_n)=t(\tilde e_1)} \phi^{p_1}_{\tilde e_1} \ldots \phi^{p_n}_{\tilde e_n} 
%$$
Moreover, the self-adjointness of $\Phi$ implies that the components $\phi_{\tilde e}^p$ satisfy:
$$
\phi_{\tilde {\bar e}}^p = (\phi_{\tilde e}^p)^* 
$$
where we recall that $\bar e$ is the edge $e$ with reversed orientation.
\end{rem}

This last Remark and its relation to the notion of R-connectedness of Definition \ref{defn:prop-R} will play a crucial role in the subsequent discussion on renormalization of the gauge field theories that correspond to $M \times F$, which we will now define.

\section{Spectral action for almost commutative manifolds}
\label{sect:sa}
Starting with an almost commutative manifold $M \times F$ with Krajewski diagram $\Gamma$ for $F$, we have now set the stage for a gauge theory on $M$. Summarizing, we have derived:
\begin{enumerate}
\item a gauge group $\SU(\A)= C^\infty(M,\SU(A_F)$ with reductive (local) gauge algebra $\su(A_F)$,
\item gauge fields $A$ in the adjoint representation of this gauge group,
\item scalar fields $\Phi$, with independent components $\phi_{\tilde e} \in \Hom( V, W)$ with $V$ and $W$ irreducible representations of $\SU(A_F)$, parametrized by the vertices $s(\tilde e)$ and $t(\tilde e)$ in the graph $\tilde \Gamma$.
\end{enumerate}
We search for gauge invariant action functionals. The most simplest, manifestly gauge invariant one is given the trace of a function of the fluctuated Dirac operator \cite{CC96, CC97}:
$$
S[A,\Phi] := \tr f\left(\frac{D+A+\gamma_5 \Phi}{\Lambda}\right),
$$
together with a real cut-off parameter $\Lambda$. 
Locally, we have
$$
D+A = i \gamma^\mu ( \nabla_\mu^S +  A_\mu).
$$
with $\nabla_\mu^S$ the spin connection on a Riemannian spin manifold $M$ and $A_\mu$ a skew-hermitian traceless matrix. The field $\Phi$ is considered as a self-adjoint element in $C^\infty(M, \End H_F)$. 

For simplicity, we take $M$ to be flat ({\it i.e.} vanishing Riemann curvature tensor) and 4-dimensional; we therefore write $\gamma_5 \equiv \gamma_M$. 
Furthermore, we will assume that $f$ is a suitable Laplace transform:
$$
f(x) = \int_{t>0} e^{-tx^2} g(t) dt.
$$

\begin{prop}[\cite{CC97}]
In the above notation, there is an asymptotic expansion (as $\Lambda \to \infty$):
\begin{equation}
\label{sa-eym}
S[A,\Phi]  \sim \sum_{m \geq 0} \Lambda^{4-m} f_{4-m} \int_M a_m (x,(D+A + \gamma_5 \Phi)^2),
\end{equation}
in terms of the Seeley--De Witt invariants of $(D+A + \gamma_5 \Phi)^2$.
The coefficients are defined by $f_k := \int t^{-k/2} g(t)dt$.
\end{prop}
Recall that the Seeley--De Witt coefficients $a_m(x,(D+A+\gamma_5 \Phi)^2)$ are gauge invariant polynomials in the fields $A_\mu$ and $\Phi$. Indeed, the Weitzenb\"ock formula gives
\begin{equation}
\label{eq:weitzenbock}
(D+A + \gamma_5 \Phi)^2 =- \nabla_\mu \nabla^\mu- \frac{1}{2} \gamma^\mu \gamma^\nu F_{\mu\nu} - \gamma_5 [D+A,\Phi] + \Phi^2
\end{equation}
in terms of the curvature $F_{\mu\nu} = \partial_\mu A_\nu - \partial_\nu A_\mu + [A_\mu,A_\nu]$ of $A_\mu$ and $\nabla_\mu= \partial_\mu + A_\mu$. Consequently, a Theorem by Gilkey \cite[Theorem 4.8.16]{Gil84} shows that (in this case) $a_m$ are polynomial gauge invariants in $F_{\mu\nu}$ and the endomorphisms
$$
E  = \frac{1}{2} \gamma^\mu \gamma^\nu F_{\mu\nu} + \gamma_5 [D+A,\Phi] - \Phi^2
$$
as well as their covariant derivatives (with respect to the connection $A_\mu$). The {\it order} $\ord$ of $a_m$ is $m$, where we set on generators:
$$
\ord A_{\mu_1; \mu_2\cdots \mu_k} = k; \qquad \ord \Phi_{;\mu_1 \cdots \mu_k} = k+1.
$$
Consequently, the curvature $F_{\mu\nu}$ has order $2$, and $F_{\mu_1 \mu_2; \mu_3 \cdots \mu_k}$ has order $k$. For example, $a_4(x,D_A^2)$ consists of terms proportional to $\tr_F F_{\mu\nu}F^{\mu\nu}$ and $\tr_F ( (\nabla_\mu \Phi)^2 + \Phi^4)$. Moreover, $a_{m} = 0$ for all odd $m$. In fact, we have:
%and more generally:
%$$
%a_{4+2k}(x,D_A^2) =- c_k \tr F_{\mu\nu} \Delta^k_A (F^{\mu\nu}) + \cO(F^3)
%$$
%for some constants $c_k$ and the Laplacian $\Delta_A= -(\partial_\mu + A_\mu)^2$ (see also \cite{Avr99} and references therein). 
%The remainder is of third and higher order in $F$, plus its covariant derivatives, adding up to an order equal to $4+2k$.
\begin{thm}
\label{thm:CC}
The spectral action for the almost commutative manifold $M \times F$ is given, asymptotically as $\Lambda \to \infty$, by
\begin{multline*}
S[A] = \frac{f_4 N \Lambda^4}{2 \pi^2} \Vol(M) - \frac{f_2\Lambda^2}{2 \pi^2}\int_M \tr_F \Phi^2
+ \frac{f(0)}{8 \pi^2} \int_M \tr_F \left( (\nabla_\mu \Phi)^2 + \Phi^4 \right)\\
  - \frac{f_0}{24 \pi^2} \int_M \tr_F F_{\mu \nu} F^{\mu\nu}  + \cO(\Lambda^{-1})
\end{multline*}
where $N = \dim \H_F$.
\end{thm}

From Remark \ref{rem:traces} it follows that we have in terms of the --- now $x$-dependent --- components $\phi_{\tilde e}^p$ of $\Phi$:
\begin{equation}
\label{sa:mass}
\int \tr_F \Phi^2 = \sum_{e,p} c_{p_1p_2}(e \bar e)  \int \tr_{s(\tilde e)}   (\phi_{\tilde e}^{p_1})^\ast \phi_{\tilde e}^{p_2}.
\end{equation}
Similarly, 
\begin{equation}
\label{sa:kinetic}
\int \tr_F (\nabla_\mu\Phi)^2 = \sum_{e,p_1,p_2} c_{p_1p_2}(e \bar e)  \int \tr_{s(\tilde e)}  (\nabla_\mu \phi_{\tilde e}^{p_1} )^*\nabla^\mu \phi_{\tilde e}^{p_2} 
\end{equation}
and finally, in terms of a sum over cycles in $\Gamma$:
\begin{align}
\label{sa:quartic}
\int \tr_F \Phi^4 &=
 \sum_{\gamma = \bar e_1 \bar e_2 e_2 e_1} \sum_{p_i} c_{p_1\ldots p_4}(\gamma) \int \tr_{s(\tilde e_{1})} (\phi^{p_4}_{\tilde e_1})^* (\phi^{p_3}_{\tilde e_2})^*   \phi^{p_{2}}_{\tilde e_{2}}  \phi^{p_{1}}_{\tilde e_{1}} \\ \nn
& \quad +\sum_{ \gamma = \bar{j(e_1)}j(e_1) \bar e_2 e_2}\sum_{p_i} c_{p_1\ldots p_4}(\gamma) \int \tr_{s(\tilde e_{2})} (\phi^{p_4}_{\tilde e_2})^* \phi^{p_{2}}_{\tilde e_{2}}   \tr_{s(\tilde e_1)} (\phi^{p_1}_{\tilde e_1})^* \phi^{p_{1}}_{\tilde e_{1}}  \\ \nn
& \quad +\sum_{ \gamma = \bar{j(e_1)} \bar e_2 j(e_1) e_2}\sum_{p_i} c_{p_1\ldots p_4}(\gamma) \int \tr_{s(\tilde e_{2})} (\phi^{p_4}_{\tilde e_2})^* \phi^{p_{2}}_{\tilde e_{2}}   \tr_{s(\tilde e_1)} (\phi^{p_1}_{\tilde e_1})^* \phi^{p_{1}}_{\tilde e_{1}}  ,
\end{align}
where $e_1, e_2$ are horizontal edges in $\Gamma$. 
These expressions will become useful later on.

\bigskip

The appearance of the Yang--Mills action and Higgs-like potential for $\Phi$ at lowest order in the spectral action on $M \times F$ is the main motivation to study this model. As a matter of fact, if we take $F$ to be described by Figure \ref{kra-sm} and choosing the $D_e$ to correspond to the physical Yukawa matrices,  then one derives in this way the full Standard Model of elementary particles, including the spontaneous symmetry breaking potential for the Higgs field \cite{CCM07, CM07}. 

In the present paper, we aim at a better understanding also of the higher-order terms in the asymptotic expansion of the spectral action and in particular the role they play as regulators of the quantum gauge theory defined at lowest order. The free part of $S[A,\Phi]$ is by definition the part of $S[A,\Phi]$ that is quadratic in the fields
\begin{equation}
\label{eq:free-sa}
S_0[A,\Phi] = \frac{1}{2} \frac{\dd}{\dd u}  \frac{\dd}{\dd v} \left( S[uA+vA,\Phi]+S[A,u\Phi+v\Phi] \right)  \bigg|_{u=v=0}.
\end{equation}
\begin{thm}
\label{thm:free-sa}
There is the following asymptotic expansion (as $\Lambda \to \infty$) for the free part of the spectral action on a flat background manifold $M$
$$
S_0[A,\Phi] \sim - \sum_{k \geq0} (-1)^k  f_{-2k} \Lambda^{-2k}  \bigg( c_k \int \tr_F \hat F^{\mu\nu} \Delta^k ( \hat F_{\mu\nu}) +  c_k'  \int \tr_F (\partial_\mu \Phi) \Delta^k (\partial_\mu \Phi) \bigg)
$$
where $\Delta$ is the Laplacian on $(M,g)$, $\hat F_{\mu\nu} = \partial_\mu A_\nu - \partial_\nu A_\mu$ and $c_k,c_k'$ are the following positive constants:
$$
c_k =  \frac{1}{8 \pi^2} \frac{(k+1)!}{(2k+3)(2k+1)!};\qquad c_k' = \frac{1}{8 \pi^2} \frac{k!}{(2k+1)!}.
$$
\end{thm}
The free Yang--Mills part was obtained in \cite{Sui11c}. The free contribution for the scalar field $\Phi$ can be derived along the same lines. Let us check the lowest order terms appearing in the above formula for $S_0[A]$ with the Yang--Mills action appearing in \cite{CCM07} (cf. Theorem \ref{thm:CC} above).
\begin{corl}
Modulo negative powers of $\Lambda$, we have
\begin{align*}
S_0 [A,\Phi] &\sim  - \frac{f_2}{4 \pi^2} \int_M  \tr_F \Phi^2 + \frac{1}{8 \pi^2}  f(0) \tr_F ( \partial_\mu \Phi)( \partial^\mu \Phi)- \frac{f_0}{24 \pi^2} \int_M  \tr_F \hat F^{\mu\nu}  \hat F_{\mu\nu} + \cO(\Lambda^{-1}).
\end{align*}
\end{corl}
We see that $S_0[A,\Phi]$ is the usual (free part of the) Yang--Mills action and a free scalar field action. In fact, we can write more concisely
$$
S_0[A,\Phi] \sim  - \frac{f_2}{4 \pi^2} \int_M  \tr_F \Phi^2 +\int_M \tr_F ( \partial_\mu \Phi) \vartheta_\Lambda(\Delta)( \partial^\mu \Phi)
- \int \tr_F \hat F_{\mu\nu} \varphi_\Lambda(\Delta) (\hat F^{\mu\nu}) 
$$
in terms of the following expansions (in $\Lambda$):
\begin{align*}
\varphi_\Lambda(x) &:= %(f_0 c_0)^{-1} 
\sum_{k \geq 0 } (-1)^k \Lambda^{-2k} f_{-2k} c_k x^k;\\
\vartheta_\Lambda(x) &:= %(f_0 c_0')^{-1} 
\sum_{k \geq 0 } (-1)^k \Lambda^{-2k} f_{-2k} c_k' x^k.
\end{align*}

This form motivates the interpretation of $S_0[A,\Phi]$ (and of $S[A,\Phi]$) as a higher-derivative gauge theory. As we will see below, this indeed regularizes the theory in such a way that $S[A ,\Phi]$ defines a superrenormalizable field theory. This comes with the usual intricacies of gauge theories with spontaneously symmetry breaking. Before proceeding with a gauge fixing and renormalization, we discuss the Higgs potential for $\Phi$. 

\subsection{Higgs mechanism and higher derivatives}
Given the above Higgs-like form of the spectral action at lowest order in the asymptotic expansion, it is natural to expand the scalar field $\Phi$ around its vacuum expectation value 
%Consider the potential term for the scalar field $\Phi$:
%$$
%V[\Phi] := S[A=0,\Phi]|_{\partial_\mu \Phi = 0} = \frac{1}{4 \pi^2} \sum_{k \geq 0} (-%1)^k \frac{f_{4-2k}}{k!} \int_M \tr_F \Phi^{2k}
%$$
%We assume that $\Phi$ attains a minimum for $V$ when 
$\langle \Phi \rangle_0 = v$, which we assume to be constant. We write 
$$
\Phi = v + \chi
$$
and refer to the fluctuations $\chi$ as the Higgs field. The constant vacuum expectation value $v$ will appear in $S[A,\Phi]$ as generating mass terms for the Higgs and the gauge field; this is spontaneous symmetry breaking (which might also not occur when $v=0$). Since the asymptotically expanded spectral action is considered as a higher-derivative theory, the interpretation of mass terms is not so straightforward. Still, we can asymptotically expand the free part of $S[A,\Phi]$ as above:
\begin{align*}
S_0[A,\Phi] & \sim \frac{1}{2}\int_M \tr_F ( \partial_\mu \chi) \vartheta_\Lambda(\Delta)( \partial^\mu \chi) + \frac{1}{2} \int_M \tr_F \chi \vartheta'_\Lambda(\Delta;v)( \chi)\\
& \quad +\frac{1}{2} \int_M (\partial_\mu A^a_\nu) \varphi_\Lambda(\Delta) (\partial^\mu A^{a\nu} - \partial^\nu A^{a\mu})  + A^{a\mu}\varphi'_\Lambda(\Delta;v)^{ab} ({A^b}_\mu)
\end{align*}
We now in addition have terms involving expansions $\vartheta'_\Lambda$ and $\varphi'_\Lambda$ which --- as $\vartheta_\Lambda$ and $\varphi_\Lambda$ do --- start with a differential operator of degree $0$ ({\it i.e.} a mass term). Besides derivatives, they also involve a series expansion in $v$.

In addition to the above free part, the splitting $\Phi=v+\chi$ induces terms in $S[A,\Phi]$ that are linear in both $A$ and $\chi$ (and in $v$):
\begin{equation}
\label{eq:cross-terms}
\int_M  \tr_F (\partial_\mu \chi)   \varpi_\Lambda(\Delta;v) ([A^\mu, v]) 
%- \int_M \tr_F(\partial_\mu \chi)   \varphi^R_\Lambda(\Delta;v) ( v A^\mu) 
\end{equation}
where we have as above a expansion defined by:
$$
\varpi_\Lambda(x;v)= %(f_0 c_0'')^{-1} 
\sum_{k \geq l \geq 0 } (-1)^k \Lambda^{-2k} f_{-2k} b_{k,l}(v) x^{k-l}
$$
and $b_{k,l}(v)$ acts (pointwise) on $\End H_F$ and is of order $2l \leq 2k$ in $v$. We also write the components of $\varpi_\Lambda$ in terms of the basis $\{e_I\}_I$ of $\Omega^1_{D_F}(\A_F) \subset \End H_F$ introduced in the previous section:
$$
\varpi_\Lambda(x;v) = \left( \varpi_\Lambda(x;v)_{IJ} \right)
$$
With a slight abuse of notation, we write $e_0$ for the identity in $\End H_F$, normalized to have Hilbert--Schmidt norm equal to 1.

For convenience, we introduce the following inner product:
\begin{equation}
\label{eq:inprod}
(\phi_1, \phi_2) = \int_M \tr_F \phi_1^* \varpi_\Lambda(\Delta;v) (\phi_2).
\end{equation}
on endomorphisms $\phi_1,\phi_2 \in C^\infty(M, \End H_F)$. Thus, the above term \eqref{eq:cross-terms} reads $(\partial_\mu \chi,[A^\mu, v]) $.

\section{$R_\xi$-gauge fixing}
\label{sect:gf}
We add a $R_\xi$-type gauge-fixing term with higher-derivatives of the following form:
\begin{align}
\label{sa-gf}
S_\gf[A,\Phi] &\sim  \frac{1}{2 \xi}  \int \tr_F \left( \partial_\mu A^{a\mu} - \xi \chi [T^a, v] \right) \varpi_\Lambda(\Delta; v) \left( \partial_\nu A^{a\nu} - \xi\chi [T^a, v] \right)%\nn \\
%&\quad- \frac{1}{2 \xi}  \int \tr_N \left( \partial_\mu A^{a\mu} + \xi \chi v T^a \right) \varphi^R_\Lambda(\Delta; v^2) \left( \partial_\nu A^{a\nu} +\xi\chi v T^a\right) 
\end{align}
which is chosen so that the terms linear in both $A$ and $\chi$ cancel the cross-terms of \eqref{eq:cross-terms}. In terms of the inner product \eqref{eq:inprod} we have more concisely:
$$
S_\gf[A,\Phi] = \frac{1}{2 \xi} \left( \partial_\nu A^{a\nu} - \xi \chi [T^a, v], \partial_\mu A^{a\mu} - \xi \chi [T^a, v] \right)
$$
where we consider $\partial_\mu A^{a\mu}(x)$ as an endomorphism of $H_F$ ({\it i.e.} as a multiple of the identity).

As usual, the above gauge fixing requires a Jacobian, conveniently described by a Faddeev--Popov ghost Lagrangian:
\begin{equation}
\label{sa-gh}
S_\gh[A,\bar C,C,\Phi] = \left(\bar C^a , \Delta C^a - \partial_\mu[A^\mu,C]^a -\xi [C, \Phi] [T^a,v]\right)
\end{equation}
Here $C,\bar C$ are the Faddeev--Popov ghost fields which are $\g$-valued fermionic fields: $C= C^a T^a$ and $\bar C = \bar C^a T^a$. Accordingly, $[C,\Phi] := C^a [T^a,\Phi]$.

\begin{prop}
The sum $S[A,\Phi] + S_\gf[A,\Phi] + S_\gh[A,\bar C, C,\Phi]$ is invariant under the BRST-transformations:
\begin{gather}
sA_\mu = \partial_\mu C + [A_\mu,C];\qquad s\Phi = -[C, \Phi]; \nn \\
s C = -\half [C,C]; \qquad s \bar C^a =  \frac{1}{\xi} \partial_\mu A^{a\mu} - \chi[T^a,v].
\label{brst}
\end{gather}
\end{prop}
\begin{proof}
First, $s(S)=0$ because of gauge invariance of $S[A,\Phi]$. Indeed, $sA_\mu$ and $s \Phi$ are just gauge transformations by the (fermionic) field $C$. 

For the gauge fixing and ghost terms, we compute 
\begin{align*}
s(S_\gf) &= \frac{1}{\xi} \left(\partial_\mu A^{a\mu} - \xi \chi [T^a, v] , -\Delta C^a + \partial_\mu[A^\mu,C]^a + \xi [C,\Phi][T^a,v] \right)
\intertext{since $s \chi = s(v+\chi) \equiv s \Phi$. On the other hand,}
s(S_\gh) &=  \left( \xi^{-1} \partial_\mu A^{a\mu} - \chi[T^a,v] , \Delta C^a - \partial_\mu[A^\mu,C]^a - \xi [C, \Phi] [T^a,v]\right)
\end{align*}
which modulo vanishing boundary terms is minus the previous expression.
\end{proof}
Note that $s^2 \neq 0$, which can be cured by standard homological methods: introduce an auxiliary (aka Nakanishi-Lautrup) field $h$ so that $\bar C$ and $h$ form a contractible pair in BRST-cohomology. In other words, we replace the above transformation in Equation \eqref{brst} on $\bar C$ by $s \bar C = -h$ and $s h = 0$. If we replace $S_\gf + S_\gh$ by $s \Psi$ with $\Psi$ an arbitrary {\it gauge fixing fermion}, it follows from gauge invariance of $S$ and nilpotency of $s$ that $S + s \Psi$ is BRST-invariant. The above special form of $S_\gf+ S_\gh$ can be recovered by choosing
$$
\Psi = ( \bar C^a, \half \xi h^a + \partial_\mu A^{a \mu} - \xi \chi[T^a,v]).
$$

We derive the propagators by inverting the non-degenerate quadratic forms in the fields $A$, $\xi$ and $C$ given by $S_0[A,\Phi] + S_\gf[A,\xi]$. 
%This can be done provided $\varphi_\Lambda, \vartheta_\Lambda$ and $\varpi_\Lambda$ are non-vanishing, which can be guaranteed by assuming that $f^{(2k)}(0) \geq 0$.
%This expression is meaningful if $f^{(2k)}(0) \geq 0$ for all $k \geq 0$. 
%Indeed, by \cite[Lemma 3]{Sui11c} we have that $f_{-2k} = (-1)^k f^{(2k)}(0)/(2k-1)!!$, so that under this assumption the summands appearing in $\varphi_\Lambda, \vartheta_\Lambda$ and $\varpi_\Lambda$ are all positive. 
This yields for the {\it gauge propagator}:
\begin{multline*}
D_{\mu\nu}^{ab}(p,v; \Lambda) = \left[ g_{\mu\nu} - \frac{p_\mu p_\nu}{p^2}\right] \left( \frac{1}{p^2 \varphi_\Lambda(p^2) + \varphi'_\Lambda(p^2;v)} \right)^{ab}  \\
+ \xi  \frac{p_\mu p_\nu}{p^2} \left( \frac{1}{p^2 \varpi_\Lambda(p^2;v)_{00} + \xi \varphi'_\Lambda(p^2;v) } \right)^{ab}.
\end{multline*}
The {\it Higgs propagator} becomes:
$$
D^{IJ}(p,v;\Lambda) = \left(\frac{1}{p^2 \vartheta_\Lambda(p^2) + \vartheta'_\Lambda(p^2;v) + \xi \mu(p^2;v) } \right)^{IJ} 
%+ \xi [T^a,v][T^b,v] \frac{1}{p^2} 
$$
where 
$$
\mu(p^2;v)_{IJ} := \tr_F e_I [T^a,v] \varpi_\Lambda(p^2;v) e_J[T^a,v]
$$
The {\it ghost propagator} is
$$
\tilde D^{ab}(p,v; \Lambda) = \frac{\delta^{ab}}{p^2 \varpi_\Lambda(p^2;v)_{00}}.
$$
%For the moment, these propagators are expansions in $\Lambda$. 

\begin{rem}
In \cite{Sui11c} we argued that for the pure Yang--Mills system the function $\varphi_\Lambda(p^2)$ appearing in the denominator of the propagator is nowhere-vanishing, provided we impose the conditions $f^{(2k)}(0) \geq 0$ on the even derivatives of $f$. Consequently, the gauge propagator did not have other poles than a physical pole at $p^2 =0$. In the present case, where we allow for spontaneous symmetry breaking, such a conclusion can not be drawn. Typically, there will be unphysical poles (involving $\xi$) appearing in the gauge and also in the Higgs and ghost propagators. Since we will be mainly concerned with renormalizability in this paper, we will ignore these poles in what follows. Of course, a treatment of (the lack of) unitarity for this higher-derivative theory does require a careful analysis of these unphysical poles as well. At lowest order (as in Theorem \ref{thm:CC}), one expects to find a cancellation of the unphysical poles appearing in the gauge and Higgs propagator, similar to \cite{Hoo71b}.
\end{rem}

\section{Renormalization on an almost commutative manifold}
\label{sect:pc}
As said, we consider the asymptotic expansion (as $\Lambda \to \infty$) of the spectral action on the AC manifold $M \times F$ as a higher-derivative field theory. This means that we will use the higher derivatives of $F_{\mu\nu}$ and $\Phi$ that appear in the asymptotic expansion as natural regulators of the theory, similar to \cite{Sla71,Sla72b} (see also \cite[Sect. 4.4]{FS80}). However, note that the regularizing terms are already present in the asymptotic expansion of the spectral action and need not be introduced as such. 
Let us consider the expansion of Theorem \ref{thm:free-sa} up to order $n$ (which we assume to be at least $4$), {\it i.e.} we set $f_{4-m} = 0$ for all $m > n$ while $f_4, \ldots f_{4-n} \neq 0$. Also, assume a gauge fixing of the form \eqref{sa-gf} and \eqref{sa-gh}. 

\begin{rem}
Note that for $n=4$ the asymptotically expanded spectral action is given by the action appearing in \ref{thm:CC} and strictly speaking not a higher-derivative gauge theory. However, in what follows it is convenient to also consider the case $n=4$, giving us the physically interesting Lagrangian.
\end{rem}

Then, we easily derive from the structure of $\varphi_\Lambda(p^2), \vartheta_\Lambda(p^2)$ and $\varpi_\Lambda(p^2;v)$ that the propagators of the gauge field, the Higgs field $\chi$, and the ghost field, respectively behave as $|p|^{-n+2}$ as $|p| \to \infty$. Indeed, in this case:
\begin{gather*}
\varphi_\Lambda(p^2) = \sum_{k=0}^{n/2-2} \Lambda^{-2k} f_{-2k} c_k p^{2k}; \quad
\vartheta_\Lambda(p^2) = \sum_{k=0}^{n/2-2} \Lambda^{-2k} f_{-2k} c_k' p^{2k}; \\
\varpi_\Lambda(p^2) = \sum_{k=0}^{n/2-2} \Lambda^{-2k} f_{-2k} c_k''(v) p^{2k}
\end{gather*}
which behave like $|p|^{n-4}$ as $|p| \to \infty$. Moreover, $\vartheta'_\Lambda(p^2;v)$ and $\varphi'_\Lambda(p^2;v)$ are subleading in $|p|$ since they behave as $v^2 |p|^{n-2}$ as $|p| \to \infty$.

Let us now consider the weights on the vertices in a Feynman graph (not to be confused with a Krajewski diagram). For gauge-Higgs interactions involving $i$ gauge and $j$ Higgs fields, the maximal number of derivatives is $n-i-j$, essentially because the total order of the corresponding term in the Lagrangian is less than or equal to $n$. Similarly, for the gauge-ghost interaction, the maximal number of derivatives is $n-3$. Finally, the Higgs-ghost interaction behaves slightly better and has $\leq n-4$ derivatives. We adopt the following notation:
\begin{center}
\begin{tabular}{|l|l||l|l|}
\hline
& number of ...& & number of ...\\
\hline\hline
$I_A$& internal gauge lines & $V_{ij}$ & gauge-Higgs vertices \\
$I_\chi$& internal Higgs lines & $\tilde V_A$ & gauge-ghost vertices\\
$\tilde I$ & internal ghost lines & $\tilde V_\chi$ & Higgs-ghost vertices\\
\hline
\end{tabular}
\end{center}
%We indicate the number of internal gauge, Higgs and ghost lines by $I^A, I^\chi, \tilde I$; the number of vertices with $i$ gauge and $j$ Higgs fields is denoted by $v_{ij}$, the number of gauge-ghost vertices is $\tilde v^A$ and finally the number of Higgs-ghost vertices is $\tilde v^\chi$.

Let us now find an expression for the {\it superficial degree of divergence} $\omega$ of a Feynman graph. In 4 dimensions, we find in terms of the above notation at loop order $L$:
$$
\omega \leq 4L - I_A(n-2) - I_\chi(n-2) - \tilde I (n-2) + 
\sum_{i+j=3}^n V_{ij} (n-i - j) + \tilde V_A (n-3)+ \tilde V_\chi (n-4).
$$
\begin{lma}
Let $E_A,E_\chi$ and $\tilde E$ denote the number of external gauge, Higgs and ghost edges, respectively. The superficial degree of divergence of the Feynman graph satisfies
$$
\omega \leq (4-n)(L-1) + 4 - (E_A+E_\chi+\tilde E).
$$
\end{lma}
\begin{proof}
We use the relations
$$
2 I_A + E_A = \sum_{i+j=3}^n i V_{ij} + \tilde V_A; \qquad
2 I_\chi + E_\chi = \sum_{i+j=3}^n j V_{ij} + \tilde V_\chi; \qquad
2 \tilde I + \tilde E = 2\tilde V_A + 2 \tilde V_\chi.
$$
Indeed, these formulas count the number of half (gauge/Higgs/ghost) edges in a Feynman graph in two ways: from the number of edges and from the valences of the vertices. We use them to substitute for $2I_A, 2I_\chi$ and $2\tilde I$ in the above expression for $\omega$ so as to obtain
$$
\omega \leq 4L - I_An -I_\chi n - \tilde I n + n \left(\sum_{i,j} V_{ij} + \tilde V_A + \tilde V_\chi \right) - (E_A + E_\chi + \tilde E)
$$
from which the result follows at once from Euler's formula $L= I_A +I_\chi + \tilde I - \sum_{i,j} V_{ij} - \tilde V_A -\tilde V_\chi +1$.
\end{proof}
As a consequence, $\omega< 0$ if $E+ \tilde E >4$ so that the theory is powercounting renormalizable. Moreover, if $n \geq 8$ then $\omega < 0$ for all $L \geq 2$: all Feynman graphs are finite at loop order greater than 1. In this case, all divergent graphs are at one loop, and satisfy $E+ \tilde E \leq 4$. We conclude that the asymptotically expanded spectral action on an AC manifold is renormalizable, and if $n \geq 8$ then it is superrenormalizable. 

Of course, the spectral action on an AC manifold being a gauge theory, there is more to renormalizability than just power counting: we have to establish gauge invariance of the counterterms. 
We already know that the counterterms needed to render the perturbative quantization of the asymptotically expanded spectral action finite are of order $4$ or less in the fields and arise only from one-loop graphs. The key property of the effective action at one loop is that it is supposed to be BRST-invariant, $s(\Gamma_{1}) = 0$.
In particular, assuming a regularization compatible with gauge invariance, the divergent part $\Gamma_{1,\infty}$ is BRST-invariant. 
We will use results from \cite{Dix91, DTV85, DTV85b, BDK90,DHTV91} on BRST-cohomology for Yang--Mills type theories to determine the form of the BRST-closed functionals of order 4 or less in the fields. In fact, in these references a relation is established between BRST-cohomology and Lie algebra cohomology for the gauge group: BRST-closed functionals are given by integrals of gauge invariant polynomials in the fields.

First, recall that with respect to the orthogonal decomposition of $\su(A_F)$ of Proposition \ref{prop:su} we can write the curvature $F_{\mu\nu} = \sum_i F^i_{\mu\nu}$ with $F_{\mu\nu}^i \in \u(k_i, \F_i)$. Gauge invariant functionals are then given by
\begin{equation}
\label{ct:F}
\int \tr F^i_{\mu\nu} F^{i \mu\nu},
\end{equation}
for all $i$. These terms appear, though with a common pre-factor, in Theorem \ref{thm:CC}.

Let us then consider the field $\Phi$ with independent components $\phi_{\tilde e}^p$, as labeled by the edges of the graph $\tilde \Gamma$ introduced at the end of Section \ref{sect:class}. The index $p$ runs from $1, \ldots, \dim S_{\tilde e}$ and the field $\phi_{\tilde e}$ is in the representation of $\su(A_F)$ induced by the irreducible representations that are given by $s(\tilde e)$ and $t(\tilde e)$ in $\tilde \Gamma^{(0)}$. The most general form of a gauge invariant functional in the components of $\Phi$ of degree 2 is given by
\begin{equation}
\label{ct:mass}
\int \tr (\phi_{\tilde e}^{p_1})^*\phi_{\tilde e}^{p_2},
\end{equation}
for all $\tilde e,p_1,p_2$. Note that these terms appear in Theorem \ref{thm:CC} (cf. Equation \eqref{sa:mass}). There is also a term of second order in $\Phi$ involving covariant derivatives, it is:
\begin{equation}
\label{ct:kinetic}
\int \tr (\nabla_\mu \phi_{\tilde e}^{p_1})^* \nabla^\mu \phi_{\tilde e}^{p_2},
\end{equation}
and is also present in Theorem \ref{thm:CC} (cf. Equation \eqref{sa:kinetic}).

Slightly more complicated is the search for the gauge invariant functionals that are quartic in the fields $\phi_{\tilde e}^p$. In terms of the graph $\tilde \Gamma$, they are given by a combination of the following sums over cycles in $\tilde\Gamma$:
\begin{gather}
\label{eq:quartic}
\sum_{\tilde\gamma= \tilde e_1 \tilde e_2 \tilde e_3 \tilde e_4} \tr_{s(\tilde e_4)= t(\tilde e_1)} \phi_{\tilde e_1}^{p_1} \cdots \phi_{\tilde e_4}^{p_4};\qquad 
\sum_{\tilde\gamma= \bar {\tilde e} \tilde e} \tr_{s(\tilde e)= t(\bar{\tilde e})} (\phi_{\tilde e}^{p_1})^* \phi_{\tilde e}^{p_2}
\sum_{\tilde\gamma'= \bar {\tilde e}' \tilde e'} \tr_{s(\tilde e')= t(\bar {\tilde e}')} ( \phi_{\tilde e'}^{p_1'})^* \phi_{\tilde e'}^{p'_2}
\end{gather}
That is, all gauge invariant quartic polynomials in $\Phi$ arise by taking traces along cycles of length 4 in $\tilde \Gamma$, and traces along pairs of cycles of total length 4. In the latter case, we exclude the possibility that the cycles $\tilde\gamma$ and $\tilde \gamma'$ both connect to the vertex ${\bf 1}$ or the vertex ${\bf \bar 1}$. In fact, the contribution arising from such cycles can be written as the trace along a single cycle of length 4, due to the fact that $\tr_{1} : \C \to \C$ acts as the identity. 

\begin{ex}
Consider the following graph $\tilde \Gamma$:
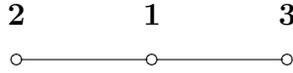
\begin{center}
\vspace{-10mm}
\begin{tabular}{c}
\begin{xy} 0;<3mm,0mm>:<0mm,3mm>::0;0,
,(11,-3)*{{\bf 2}}
,(17,-3)*{{\bf 1}}
,(23,-3)*{{\bf 3}}
,(11,-5)*\cir(0.3,0){}
,(17,-5)*\cir(0.3,0){}
,(23,-5)*\cir(0.3,0){}
,(11.2,-5);(16.8,-5)**\dir{-}
,(17.2,-5);(22.8,-5)**\dir{-}
%,(11.1,-4.9);(22.9,-4.9)**\crv{(17,-2.5)}
\end{xy}
\end{tabular}
\end{center}
The pair of cycles $({\bf 21})({\bf 12})$ and  $({\bf 13})({\bf 31})$ give a contribution 
$$
\tr \phi_{(\bf 21)}\phi_{(\bf 12)} \tr \phi_{(\bf 31)}\phi_{(\bf 13)} = (\phi_{(\bf 21)}\phi_{(\bf 12)})( \phi_{(\bf 31)}\phi_{(\bf 13)})
$$
since $(\phi_{(\bf 21)}\phi_{(\bf 12)})$ and $(\phi_{(\bf 31)}\phi_{(\bf 13)})$ are elements in $\Hom(\C,\C) \simeq \C$. Now, the concatenated cycle $({\bf 21})({\bf 12})({\bf 13})({\bf 31})$ of length 4 gives the same contribution 
$$
\tr \phi_{(\bf 21)}\phi_{(\bf 12)} \phi_{(\bf 31)}\phi_{(\bf 13)} = (\phi_{(\bf 21)}\phi_{(\bf 12)})( \phi_{(\bf 31)}\phi_{(\bf 13)})
$$
for the same reason.
\end{ex}

Now, adopting Definition \ref{defn:prop-R}, if the Krajewski diagram $\Gamma$ is R-connected (in dimension 4) the above traces can always be written in terms of cycles of length 4 in $\Gamma$ which are precisely the terms that are present in Theorem \ref{thm:CC} (cf. Equation \eqref{sa:quartic}). We conclude:
\begin{thm}
Let $M \times F$ be an almost commutative manifold with $\dim M = 4$; suppose that $A_F$ contains no copies of $\R$, and either no complex subalgebras, or at least one non-trivial complex subalgebra (cf. Proposition \ref{prop:su}). Consider the asymptotically expanded spectral action up to order $n \geq 4$. 

If the Krajewski diagram describing the finite real spectral triple for $F$ is R-connected in dimension 4, then the asymptotically expanded spectral action (with $f_{4-m}=0$ for all $m > n$) for $M \times F$ is renormalizable as a gauge theory. Moreover, it is superrenormalizable as a gauge theory if $n \geq 8$. 
\end{thm}
As a corollary, we find that the asymptotically expanded Yang--Mills spectral action is superrenormalizable ($n \geq 8$), as was previously shown in \cite{Sui11b,Sui11c}. Indeed, the spectral triple of Example \ref{ex:ym-st} has Krajewski diagram
\begin{center}
\vspace{-10mm}
\begin{tabular}{c}
\begin{xy} 0;<3mm,0mm>:<0mm,3mm>::0;0,
,(11,-2.5)*{{\bf N}}
,(8.5,-5)*{{\bf N}}
,(11,-5)*\cir(0.3,0){}
\end{xy}
\end{tabular}
\end{center}
which is R-connected in a trivial way.

Similarly, Proposition \ref{prop:kra-sm} implies that the asymptotically expanded spectral action that at lowest order is the Standard Model, is renormalizable as a gauge theory. In particular, choosing $n=4$ this implies that the Standard Model is renormalizable as a gauge theory.

%\begin{rem}
%The above argument can be applied also to the action given by the lowest-order terms appearing in Theorem \ref{thm:CC}. Of course, this action is not power-counting superrenormalizable, but it is power-counting renormalizable. The above BRST-invariance of one-loop counterterms is generalized to an order-by-order in the loop number argument, adopting eg. the Zinn--Justin equation. One then arrives at the conclusion that if the Krajewski diagram for the underlying almost commutative manifold $M \times F$ is R-connected (dimension 4), then the action of Theorem \ref{thm:CC} is renormalizable as a gauge theory.
%\end{rem}

\begin{ex}
Let us illustrate the possible failure of renormalizability for the Krajewski diagram in Example \ref{ex:not-R} which is not R-connected. There are fields $\phi_{({\bf 12})}$ and $\phi_{({\bf \bar 1 3})}$ that could combine to give a gauge-invariant counterterms proportional to
$$
\left( (\phi_{({\bf 12})})^* \phi_{({\bf 12})} \right) \left((\phi_{({\bf \bar 1 3})})^* \phi_{({\bf \bar 1 3})}\right)
$$
However, this term can never appear in the asymptotic expansion of the spectral action, since the edges $({\bf 12}),({\bf 21}), ({\bf \bar 1 3})$ and $({\bf 3 \bar 1})$ do not lift to a cycle in $\Gamma$. 
\end{ex}

Note that the asymptotically expanded spectral action is not {\it multiplicatively} renormalizable, since the coefficients in front of the counterterms might be different for different indices (such as $i$, $\tilde e$, and $p$). 
This is in contrast with the classical action in Theorem \ref{thm:CC} where there is a typical {\it unification} of couplings for all simple factors of the gauge group. This suggests that one takes the spectral action $S[A,\Phi]$ (plus gauge fixing) as a starting point for the renormalization group flow to then run the action to arbitrary energy scales. 

It is an open question whether this approach to renormalizing the asymptotically expanded spectral action using the intrinsic higher-derivative regulators is equivalent to perturbatively quantizing the gauge theory defined by the lowest-order terms (appearing in Theorem \ref{thm:CC}) using, say, dimensional regularization and minimal subtraction. Evidence that this might be true can be found in \cite{MRR95b} and is the subject of future research.

%\bibliography{references}
\newcommand{\noopsort}[1]{}\def\cprime{$'$}

\end{document}